\pdfoutput=1

\documentclass[preprint]{elsarticle}

\usepackage{amssymb}
\usepackage{amsmath}
\usepackage{amsthm}
\newtheorem{theorem}{Theorem}

\usepackage[figuresright]{rotating}

\newcommand{\rmd}{\mathrm{d}}

\begin{document}

\begin{frontmatter}





  \title{Effect of Bitcoin fee on transaction-confirmation process}


\author[naist]{Shoji Kasahara}
\ead{kasahara@ieee.org}

\author[naist]{Jun Kawahara}

\address[naist]{Graduate School of Information Science, Nara Institute
of Science and Technology\\
Takayama 8916-5, Ikoma, 6300192 Nara, Japan}

\begin{abstract}
  In Bitcoin system, transactions are prioritized according to
  transaction fees. Transactions without fees are given low priority
  and likely to wait for confirmation. Because the demand of micro
  payment in Bitcoin is expected to increase due to low remittance
  cost, it is important to quantitatively investigate how transactions
  with small fees of Bitcoin affect the transaction-confirmation
  time. In this paper, we analyze the transaction-confirmation time by
  queueing theory.  We model the transaction-confirmation process of
  Bitcoin as a priority queueing system with batch service, deriving
  the mean transaction-confirmation time. Numerical examples show how
  the demand of transactions with low fees affects the
  transaction-confirmation time. We also consider the effect of the
  maximum block size on the transaction-confirmation time.
\end{abstract}


\end{frontmatter}


\section{Introduction}
\label{sec:intro}

Bitcoin is a digital currency system that was invented by Satoshi
Nakamoto in 2008 \cite{Nakamoto08}. Unlike the existing online payment
systems such as credit cards and debit ones, a remarkable feature of
Bitcoin system is its decentralized nature. Bitcoin does not have a
central authority to manage Bitcoin transactions. All the Bitcoin
transactions are registered in the ledger called \textit{blockchain},
and the blockchain is maintained by a volunteer-based peer-to-peer
(P2P) network. Volunteer nodes joining the P2P network hold the same
replica of the blockchain, which enables everyone to check consistency
of transactions.

From information-technology point of view, Bitcoin is fast, secure,
and has lower fees than the existing payment schemes. These features
make Bitcoin advantageous for both consumers and retailers. Due to low
fees for processing transactions, Bitcoin is expected to accelerate
the use of micro payment such as buying daily items and small amount
remittance.


Every transaction needs to be stored into a \textit{block} that a
volunteer node, called \textit{miner}, creates, and the block needs to
be appended to the tail of the blockchain when a miner succeeds in
creating the block.  An average time interval of block creation is
adjusted to be about 10 minutes.  Only transactions in blocks included
in the blockchain are admitted as valid ones, which are called
\textit{confirmed} transactions.  It is said that to avoid
double-spending, if the user receives the coin in a transaction,
he/she should wait to use it until the block including the transaction
and some subsequent blocks are created \cite{confirmation}.

When a sender creates a transaction, he/she can make the transaction
include a fee, which can be received by the miner who creates a block
that stores the transaction. There is no incentive for miners to store
transactions without fee into the block they are creating. Since the
Bitcoin system restricts the number of transactions which a block can
hold, miners may put a higher priority on transactions with a larger
fee. Therefore, the transaction-confirmation time of transactions
with a small fee tend to be much larger than those of ones with a
large fee.

Note that in micro-payment case, the fee amount of micro-payment
transactions is likely to be small due to its small amount remittance.
If the use of micro payment becomes popular in the future, the
confirmation time of transactions with small fees will be too long for
users to make micro payment.



In this paper, we consider how the growth of micro payment affects the
confirmation process of small amount transactions. We collect statistics
from the blockchain, investigating the transaction-confirmation
time. Then we model the transaction-confirmation process of Bitcoin as a
queueing system with bulk service and priority mechanism, deriving the
mean transaction-confirmation time for each-priority transaction.  In
numerical examples, we show how the transaction-confirmation time is
affected with the increase in demand of micro payment.

The rest of the paper is organized as follows. We briefly review the
related work in Section \ref{sec:relatedwork}. Section
\ref{sec:summary} shows a summary of Bitcoin system, mainly focusing
on the blockchain construction and the impact of transaction fee on
the transaction-confirmation process.  Section \ref{sec:statistics}
shows some statistics about Bitcoin, some of which are used in the
later experiments. In Section \ref{sec:analysis}, we describe the
queueing model for the confirmation process of Bitcoin system, and the
analysis of the queueing model is presented. In Section
\ref{sec:numericalexamples}, we show some numerical examples,
discussing the effect of the demand of transactions with low fee on
the confirmation time of transactions. Concluding remarks are given in
Section \ref{sec:conclusion}.

\section{Related Work}
\label{sec:relatedwork}

Recently, Bitcoin has attracted considerable attention, and been
widely studied in various research communities. For example, the
economic community studies Bitcoin system from the virtual-currency
point of view. The aspects of applications of encryption and P2P
networking are of interest in computer science \cite{DW13}. The
community of social science focuses on the incentive mechanism of
Bitcoin ecosystem. Comprehensive reviews in terms of technology
principles, history, risks and regulatory issues are well provided in
\cite{BCEM15,TS16,BMCetal15}. Almost all papers on Bitcoin are
introduced in \cite{btcresearch}. Here, we present only papers close
to our research.

One of important issues in Bitcoin is transaction fee. It is expected
that transaction fees become incentives for miners to provide much
computation power in order to verify transactions.  The authors of
\cite{MB15} investigate the trends of transaction fees by analyzing
55.5 million transaction records, revealing the regime shift of
Bitcoin transaction fees. It is shown that transactions with non-zero
fee are likely to be processed faster than those with zero fee, and
that the amount of fee doesn't affect the transaction latency
significantly. In terms of the latter claim, however, their
statistical analysis shows the tendency that transactions with small
fee are likely to wait longer than those with large fee.

Another important issue is the maximum block size. Currently, the
maximum block size is limited to 1 Mbyte due to a security reason of
spam attack \cite{PeterR15}.  It is reported in \cite{scalability}
that Bitcoin handles at most seven transactions per second (tps) due
to the maximum block size of 1 Mbyte. In order for Bitcoin to scale to
tens of thousands of tps, which is equivalent to the processing speed
of credit card transactions, enlarging the maximum block size is
considered. There exist many discussions about the effect of the
maximum block size on the incentive of miners. To the best of the
authors' knowledge, however, there is no work for quantitatively
investigating the impact of the enlargement of the maximum block size
on the transaction-confirmation time.

Block confirmation time also affects the scalability of Bitcoin.
Sompolinsky and Zohar \cite{SZ13,SZ15} propose a modification to the
blockchain, called GHOST, so that block confirmation time becomes
about 600 times shorter than the original Bitcoin without loosing the
security of Bitcoin.  Kiayias and Panagiotakos \cite{KP15} show a formal
security proof and the speed-security tradeoff of GHOST. A security
issue about shortening block confirmation time is double spending,
which is studied in \cite{BDEWW13,KAC12}.


The block-construction process can be modeled as a queueing system
with batch service, in which a group of customers leave the system
simultaneously at service completion.  There exist literature for the
analysis of queues with batch service.  Chaudhry and Templeton
consider an M/$\rm G^{\rm B}$/1 queueing system with batch service
\cite{CT81,CT83}. In M/$\rm G^{\rm B}$/1, customers arrive at the
system according to a Poisson process, the number of servers is one,
and the service time distribution follows a general distribution. If
there exist customers in queue at a service completion, the server
accommodates customers as a batch, where the batch size is limited to
some constant. Using supplemental variable technique, the authors
derive the joint distribution of the remaining service time and the
number of customers in queue. In \cite{CT81,CT83}, however, the
priority mechanism is not taken into consideration. To the best of the
authors' knowledge, priority queueing system with batch service has not
been fully studied yet.  In this paper, we model the
transaction-confirmation process of Bitcoin as a queueing system in
which both priority mechanism and batch service are taken into
consideration.

\section{Summary of Bitcoin System}
\label{sec:summary}

In this section, we give a brief summary of Bitcoin system.  The
readers are referred to \cite{Antonopoulos14} for details.

\subsection{Transaction Confirmation Process}
\label{sec:summary:subsec:tcp}

The Bitcoin system realizes virtual currency with two types of
information data: \textit{transactions} and \textit{blocks}.  A
transaction is the base of value transfer between payer and payee,
while a block is a data unit for storing several confirmed
transactions.

When a payer makes payment to a payee in Bitcoin system, the payer
issues a transaction into the Bitcoin P2P network. The transaction
contains the amount of payment, the source account(s) of the payer,
the destination account(s) of the payee, and the fee that the payer
pays to a miner (and the others). The transaction is propagated
through the P2P network, and temporally stored in memory pool of
volunteer nodes, called \textit{miners}.

The role of miner nodes is to generate a block, which contains
transactions to be validated.  Miners try to solve a mathematical
problem based on a cryptographic hash algorithm for block generation
(referred to as \textit{proof-of-work} \cite{Nakamoto08}).  The miner
who finds its solution first becomes a winner, and is awarded
reward\footnote{In 2017, the output of the coinbase for one-block
  mining is 12.5 bitcoin. The output value of it is halved every
  210,000 blocks. Since the mining time for one block is 10 minutes on
  average, this corresponds to a four-year halving schedule.}, which
consists of some fixed value called coinbase and the fees of
transactions included in the block, and the right to add a new block
to the blockchain.  The solution to the problem is included in the new
block, and the miners try again to solve a new mathematical problem
for the next block. This competition process is called
\textit{mining}. Embedding the solution for the current block to the
next block plays an important role for preventing from falsification
of previous blocks.  The difficulty of problems in Bitcoin mining is
automatically adjusted by the system so that the time interval between
consecutive block generations is 10 minutes on average.

Since miners do not gain any profit from transactions without fee and
the total size of transactions that a block can store is limited to
1 Mbyte in the Bitcoin system, some miners may ignore such
transactions. Therefore, it is considered that the
transaction-confirmation time of transactions without fee is much
bigger than those of ones with fee.

In \cite{MB15}, the authors study trends of Bitcoin transaction fee
conventions by analyzing the transaction fees paid with 55.5 million
transactions recorded in the blockchain. They find that the
confirmation time of transactions without fee are longer than those
with fee. It is also reported that difference between
transaction-confirmation times for different fees are not
significant. In terms of the latter claim, however, their statistical
analysis reveals that transactions with fee of 0.0005 are likely to
wait longer than those with fee of 0.001. (See Table 2 in
\cite{MB15}.)  If the demand of transactions for micro payment
increases in future, those transactions may suffer from a very long
confirmation time because payers of micro payment are not willing to
pay fee and the resulting priority of their transactions is low.

\section{Bitcoin Transaction Statistics}
\label{sec:statistics}

In this section, we show some statistics of Bitcoin blocks and
transactions.  We collected data of blocks and transactions from
\verb|blockchain.info| \cite{blockchaininfo}. We chose the two-year
mining period from October 2013 to September 2015.

\subsection{Basic Statistics}
\label{sec:statistics:subsec:bs}

\begin{table}[t]
  \centering
  \caption{Block-generation time.}
  \label{tab:statistics:bgt}
  \begin{tabular}{lr}\hline
    Mean  [s] & 544.09 \\
    Variance & $2.9277 \times 10^{5}$ \\
    Maximum  [s] & 6,524 \\
    Minimum  [s] & 0 \\
    Median  [s] & 377 \\\hline
  \end{tabular}
\end{table}

Table \ref{tab:statistics:bgt} shows statistics of block-generation
time. The statistics are calculated from 115,921 blocks in the
measurement period.  In this table, the mean block-generation time is
544.09 s, approximately 9 minutes.  This is smaller than 10 minutes,
the average time interval between consecutive block generations. This
result, however, supports that Bitcoin mining is managed according to
the system protocol.

\begin{table}[t]
  \centering
  \caption{Number of transactions in a block.}
  \label{tab:statistics:nt}
  \begin{tabular}{lr}\hline
    Mean [transactions] & 529.27 \\
    Variance & $2.5152 \times 10^5$ \\
    Maximum [transactions] & 12,239 \\
    Minimum [transactions] & 0 \\
    Median [transactions] & 386 \\\hline
  \end{tabular}
\end{table}

Table \ref{tab:statistics:nt} shows the number of transactions in a
block. Here, we count not only transactions issued by users, but also
coinbase transactions. The mean number of transactions in a block is
529.27, and hence the mean rate of transaction processing is 1.05
transaction/s.

\begin{table}[t]
  \centering
  \caption{Transaction size in byte.}
  \label{tab:statistics:ts}
  \begin{tabular}{lr}\hline
    Mean  & 571.34 \\
    Variance & $3.7445\times 10^6$\\
    Maximum & 999657 \\
    Minimum & 62 \\
    Median & 259 \\\hline
  \end{tabular}
\end{table}

Table \ref{tab:statistics:ts} shows the statistics of the transaction
size in byte. The mean transaction size for the two-year period is
571.34 bytes. Since the maximum block size is 1 Mbyte, we can roughly
approximate the maximum number of transactions in a block equal to
1750.3.

\begin{table}[t]
  \centering
  \caption{Cumulative frequency of fee amount for transactions.}
  \label{tab:statistics:cfrat}
  \begin{tabular}{lr}\hline
    BTC & Frequency \\\hline
    0 & 1378501 \\
    0.00001 & 3050709 \\
    0.0001 & 42881857 \\
    0.001  & 60723356 \\
    0.01   & 61219997 \\
    0.1  & 61236481 \\
    1 & 61236972 \\
    10 & 61237045 \\ \hline
  \end{tabular}
\end{table}

In order to investigate the impact of transactions with small fee on
the transaction-confirmation time, we classify transactions into
priority classes. Remind that the confirmation time of transactions
with a small fee are longer than those with a large fee \cite{MB15}.
This implies that transactions without fee are given the lowest
priority for the block-inclusion process. Therefore, we classify
transactions into two types, high (H) and low (L), in terms of the
amount of fee added to a transaction. Transactions with fee greater
than or equal to 0.0001 BTC\footnote{In May 2017, 0.0001 BTC is about
  0.12 USD.} are classified into H class, while those without fee
smaller than 0.0001 BTC are prioritized as L class. We show the
cumulative frequency of the fee amount for transactions in Table
\ref{tab:statistics:cfrat}.

\begin{table}[t]
  \centering
  \caption{Transaction-type statistics.}
  \label{tab:statistics:tct}
  \begin{tabular}{lrrr}\hline
    Statistic     & Classless & H & L \\\hline
    Number of transactions & 61,353,014 & 57,058,947 & 4,294,067 \\
    Mean TCT [s] & 1075.0 & 874.13 & 3744.1 \\
    Variance of TCT & $1.8989 \times 10^8$  & $8.4505 \times 10^7$ & $1.5826 \times 10^9$ \\
    Maximum of TCT & $3.1045\times 10^7$ & $3.1045\times 10^7$ & $2.6244\times 10^7$ \\
    Minimum of TCT & 0 & 0 & 0 \\
    Median of TCT & 510 & 502 & 640 \\
    Mean arrival rate  & 0.97275 &  0.90466 & 0.068082 \\\hline 
  \end{tabular}
\end{table}

Table \ref{tab:statistics:tct} shows the statistics of transactions by
type. Here, classless indicates the statistics for all the
transactions, and TCT is the transaction-confirmation time. The mean
arrival rate is the number of transactions per day.  In this
table, the mean transaction-confirmation time for the overall
transactions is 1{,}075.0 [s] $\approx 17.917$ minutes, almost twice
greater than the mean block-generation time. 

In terms of priority-type statistics, the mean transaction-confirmation
time for L class is greater than that for H class, and its difference
is 2{,}870.0 [s] $\approx$ 47.833 minutes.


\subsection{Fee-amount distribution and transaction-arrival rate}

\begin{figure}[t]
  \centering
  \includegraphics[width=.8\textwidth,clip]{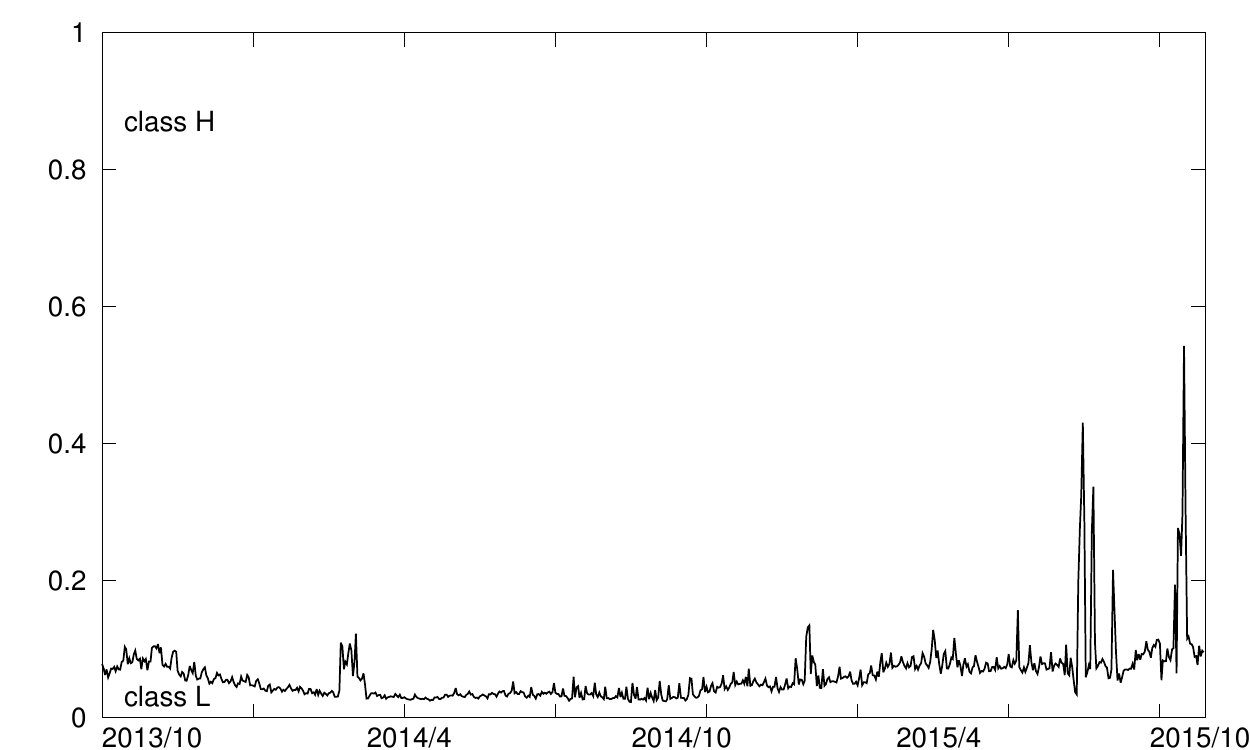}
  \caption{Trend of fee-amount distribution over time.}
  \label{fig:statistics:payment}
\end{figure}

Figure \ref{fig:statistics:payment} illustrates how the fee-amount
distribution changes over time. The fee-amount distribution is the
ratio of the amount of H/L-transactions to that of transactions issued
in one day.  Each region in Figure \ref{fig:statistics:payment} shows
the percentage of transactions in two different classes.

In Figure \ref{fig:statistics:payment}, the percentage of each class
fluctuates in a small range, except that ${\rm L}$-class has a spike
from July 2015 to October 2015. It is reported in
\cite{blockchaininfo} that the number of transactions per day exhibits
a rapid increase during the same period.  From these observations, we
can claim that the percentages of H and L classes remain almost the
same even though the volume of transactions increases rapidly.

\begin{figure}[t]
  \centering
  \includegraphics[width=.8\textwidth,clip]{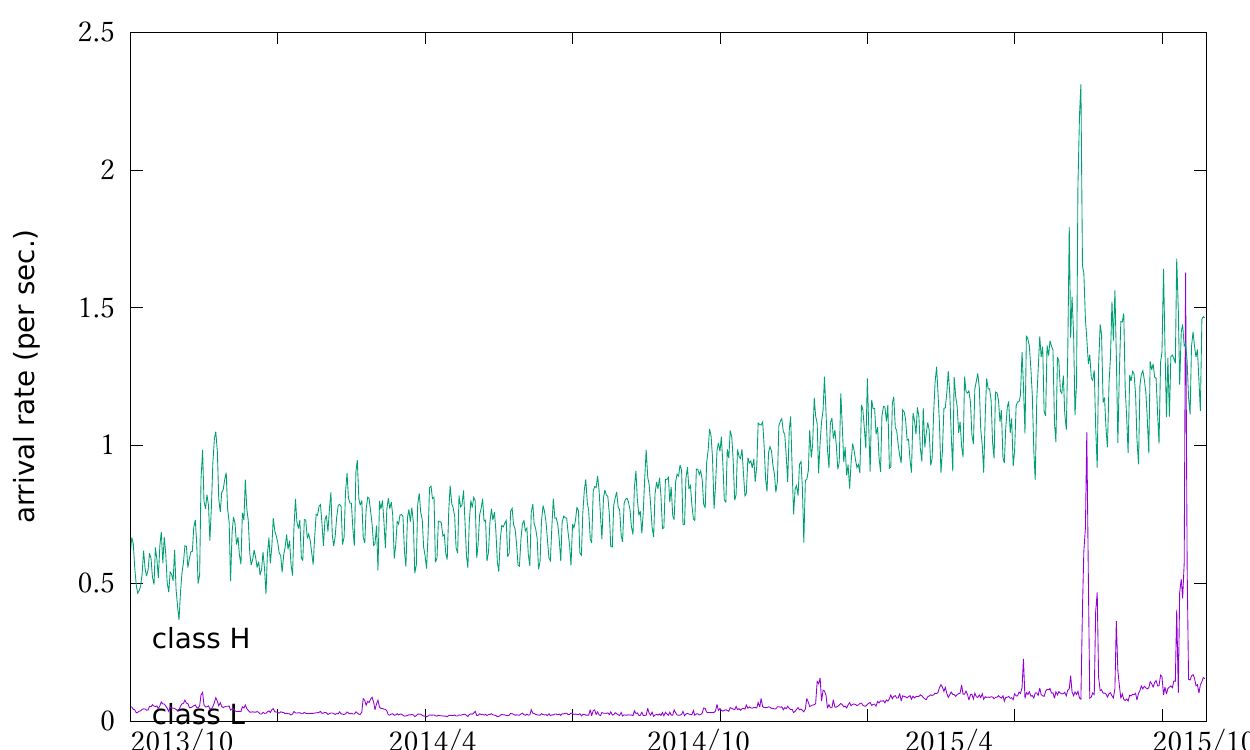}
  \caption{Trend of transaction-arrival rates of two priority
    classes.}
  \label{fig:statistics:category}
\end{figure}

Figure \ref{fig:statistics:category} represents the
transaction-arrival rate of each class. The horizontal axis is day,
and its origin is October 1, 2013. We observe in this figure that the
transaction-arrival rate of each class gradually increases with
fluctuation over time. The exceptional spikes are observed in the
range of 650 to 730, the same period in Figure
\ref{fig:statistics:payment}.

From these figures, we can expect that the transaction-arrival rate
monotonically grows, keeping the same percentage of fee-amount
class. 



\section{Priority Queueing Analysis}
\label{sec:analysis}

This section describes our queueing model of Bitcoin transaction
processing and main results of transaction-confirmation time. The
detailed derivations are presented in \ref{app:sec:analysis}.

\subsection{Mean Transaction-Confirmation Time}

Let $S_i$ denote the $i$th block-generation time.  In this paper, we
regard a block-generation time as a service time. We assume
$\{S_i\}$'s are independent and identically distributed (i.i.d.) and
have a distribution function $G(x)$. Let $g(x)$ denote the probability
density function of $G(x)$.  The mean block-generation time $E[S]$ is
given by
\[
E[S] = \int_0^\infty x \rmd G(x) = \int_0^\infty x g(x) \rmd x.
\]

A transaction arrives at the system according to a Poisson process
with rate $\lambda$. Transactions arriving to the system are served in
a batch manner. A batch service starts when a transaction arrives at
the system in idle state. The consecutive transactions arriving at the
system are served in a batch until the number of batch size equals
$b$. That is, newly arriving transactions are included into the
creating block as long as the resulting block size is smaller than the
maximum block size $b$. This assumption follows from the behavior of the
default Bitcoin client described in \cite{Antonopoulos14}.

Let $N(t)$ denote the number of transactions in system at time $t$,
and $X(t)$ denote the elapsed service time at $t$.  We define
$P_n(x,t)$ ($n=1,2,\ldots, x,t \geq 0$) and $P_0(t)$ as\footnote{We
  follow the definition of $P_n(x,t)$ in \cite{CT81,CT83}. Rigorously,
  we define $P_n(x,t)$ as
  \[
    P_n(x,t) = \frac{\rmd}{\rmd x}\Pr\{N(t)=n, X(t)\leq x \}
  \]
  assuming it exists.
}
\begin{eqnarray*}
  P_n(x,t) \rmd x &=& \Pr\{N(t)=n, x<X(t)\leq x+\rmd x \},\\
  P_0(t) &=& \Pr\{N(t)=0\}.
\end{eqnarray*}
Note that $P_n(x,t) \rmd x$ is the joint probability that at time $t$,
there are $n$ transactions in system and the elapsed service time lies
between $x$ and $x+\rmd x$. We also define limiting distributions
$P_n(x) = \lim_{t\rightarrow\infty}P_n(x,t)$ and
$P_0=\lim_{t\rightarrow\infty}P_0(t)$.

Let $\xi (x)$ denote the hazard rate of $S$, which is given by
\[
\xi (x) = \frac{g(x)}{1-G(x)}.
\]

Let $T$ denote the sojourn time of a transaction. In the context of
Bitcoin, $T$ is the transaction-confirmation time, i.e., the time
interval from the time epoch at which a user issues a transaction to
the point when the block including the transaction is confirmed. Then,
we have the following theorem.

\begin{theorem}

The mean transaction-confirmation time $E[T]$ is given by
\begin{eqnarray}
  E[T] &=& {1\over 2\lambda^2 (b-\lambda E[S])}
      \left(\rule{0pt}{18pt}
       \sum_{k=1}^b \alpha_k\left[\rule{0pt}{14pt}
         b(b-1)+\{(b+1)b-k(k-1)\}\lambda E[S] \right.\right.\nonumber\\
     & & \left.\left. +(b-k)\lambda^2E[S^2]\rule{0pt}{14pt}\right]
        - \lambda\left\{b(b-1)-\lambda^2
          E[S^2]\right\}\rule{0pt}{18pt}\right),
  \label{eq:meannumtran:01}
\end{eqnarray}
where
\[
\alpha_k = \int_0^\infty P_k(x) \xi(x) \rmd x.
\]
\end{theorem}

\begin{proof}
  See \ref{app:sec:proof:sub:theorem1}.
\end{proof}

\subsection{Transaction-Confirmation Time for Priority Queueing Model}

In this subsection, we consider the system in which transactions are
prioritized for the inclusion to a block, deriving the mean
transaction-confirmation time for each priority class.

We assume that transactions are classified into $c$ priority
classes. For $1\leq i,j \leq c$, $i$ class transactions have priority
over transactions of class $j$ when $i<j$.  Let $\lambda_i$
($i=1,2,\ldots,c$) denote the arrival rate of $i$-class
transactions. We assume that $\sum_{i=1}^c \lambda_i E[S]<1$.  We
define $T_i$ as the sojourn time of class $i$ transactions. For
simplicity, we introduce the following notation
\[
\overline{\lambda}_i = \sum_{k=1}^i \lambda_k,\quad i=2,3,\ldots,c.
\]

Assuming that the system is work conserving, we have the following
theorem.
\begin{theorem}
  \label{th:priority}
  Let $T_i$ ($i=1,\ldots,c$) denote the confirmation time of class $i$
  transactions. 
  \[
    E[T_1]=f(\lambda_1),
  \]
  \[
    E[T_i] = \frac{1}{\lambda_i}
    \left(\overline{\lambda}_i f(\overline{\lambda}_i) -
      \sum_{k=1}^{i-1} \lambda_k E[T_k]
    \right), \quad i=2,3,\ldots, c,
  \]
where $f(\lambda)=E[T]$, given by (\ref{eq:meannumtran:01}).
\end{theorem}

\begin{proof}
  See \ref{app:sec:proof:sub:priority}.
\end{proof}

\noindent {\bf Remark:} Strictly speaking, our priority queueing model
is not work conserving. (See \ref{app:sec:proof:sub:priority}.)
$E[T_i]$'s given in Theorem \ref{th:priority} are approximations which
work well for high utilization. When the block-generation time follows
an exponential distribution, however, $E[T_i]$'s in Theorem
\ref{th:priority} agree with simulation results, as shown in
subsection \ref{sec:nr:subsec:vc}.

In the following section of numerical examples, we consider two
priority-class case: high and low. Let $\lambda_H$ and $\lambda_L$
denote the arrival rate of high-priority transactions and that of
low-priority ones, respectively. Let also $T_H$ and $T_L$ denote the
sojourn time of high-priority transactions and that of low-priority
ones, respectively. In this two priority-class case, we obtain
\begin{eqnarray}
E[T_H] &=& f(\lambda_H), \label{eq:meandelay:02}\\
E[T_L] &=&
   \left(\frac{\lambda_H}{\lambda_L}+1\right)f(\lambda_H+\lambda_L) - 
     \frac{\lambda_H}{\lambda_L}f(\lambda_H). \label{eq:meandelay:03}
\end{eqnarray}

\section{Numerical Examples}
\label{sec:numericalexamples}

In this section, we show some numerical examples obtained from the
analytical results in previous section. First, we consider the
distribution of block-generation time with a simple mining
model. Then, we show the transaction-confirmation times of H- and
L-class transactions, investigating how the transaction-arrival rate
and the block size affect the performance measure.

\subsection{Distribution of block-generation time}
\label{sec:numericalexamples:sub:block-generation-time}

In order to calculate the mean transaction-confirmation time, we need
to determine $G(x)$, the distribution of the block-generation time.
In \cite{GKKT16}, the authors claim that the block-generation time is
exponentially distributed. They consider a hash calculation by a miner
node as a Bernoulli trial, which is independent of previous hash
calculations. This yields that the number of experiments for the first
success is given by geometric distribution, and hence it can be
approximated by exponential distribution.  In
\ref{app:sec:block-gener-time}, we show an alternative approach to the
block-generation time distribution with extreme value theory.

In subsection \ref{sec:statistics:subsec:bs}, we showed that the mean
block-generation time is 544.09 [s]. That is, the mean
block-generation rate is $1.8379\times 10^{-3}$.  In the following, we
assume that the block-generation time $S$ follows the exponential
distribution given by
\[
  G(x) = 1- e^{-\mu x},
\]
where $\mu=1.8379\times 10^{-3}$.

\begin{figure}[t]
  \centering
  \includegraphics[width=.8\textwidth,clip]{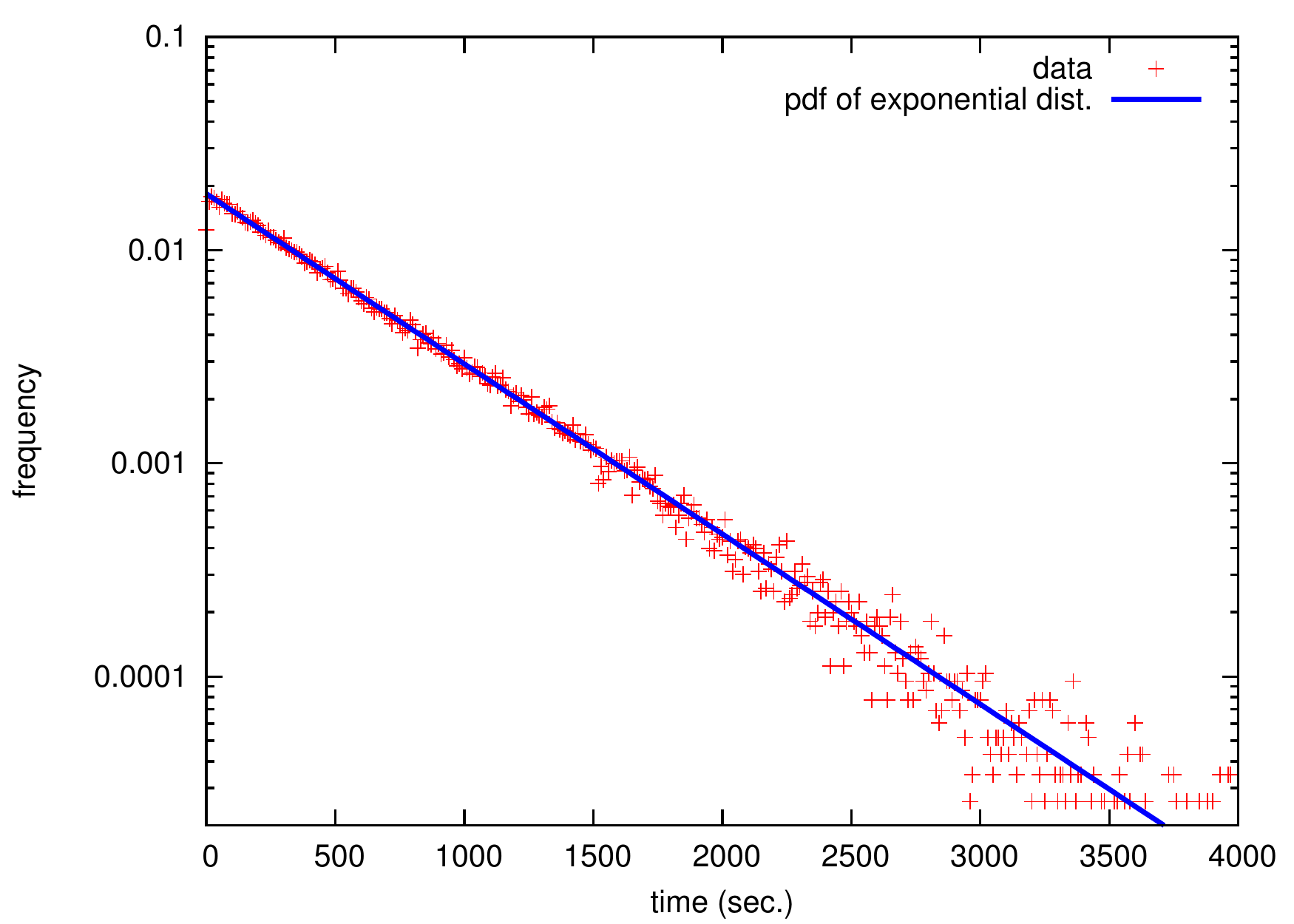}
  \caption{Relative frequency and exponential probability density
    function of block-generation time.}
  \label{fig:ne:fitting}
\end{figure}

In Figure \ref{fig:ne:fitting}, we plot the relative frequency of the
block-generation time obtained from the measured data, and the
probability density function of the above exponential
distribution. The horizontal axis represents the block-generation time
in second, and the vertical axis is the logarithmic scale of the
frequency values. This figure shows a good agreement between the
measured data and exponential distribution.

From the assumption of exponential distribution for the
block-generation time, we set $E[S]$ and $E[S^2]$ as
\[
E[S]=\frac{1}{\mu}=544.09, \quad E[S^2] = \frac{2}{\mu^2} =
5.9208\times 10^5. 
\]
The Laplace-Stieltjes transform (LST) of $G(x)$ is given by
\[
G^*(s) = \frac{\mu}{s+\mu}.
\]
With the above setting, we calculate mean sojourn times of
transactions in previous section.

\subsection{Verification and Comparison}
\label{sec:nr:subsec:vc}

\subsubsection{Verification of analysis}

\begin{figure}[t]
  \centering
  \includegraphics[bb=0 0 410 302,width=.8\textwidth,clip]{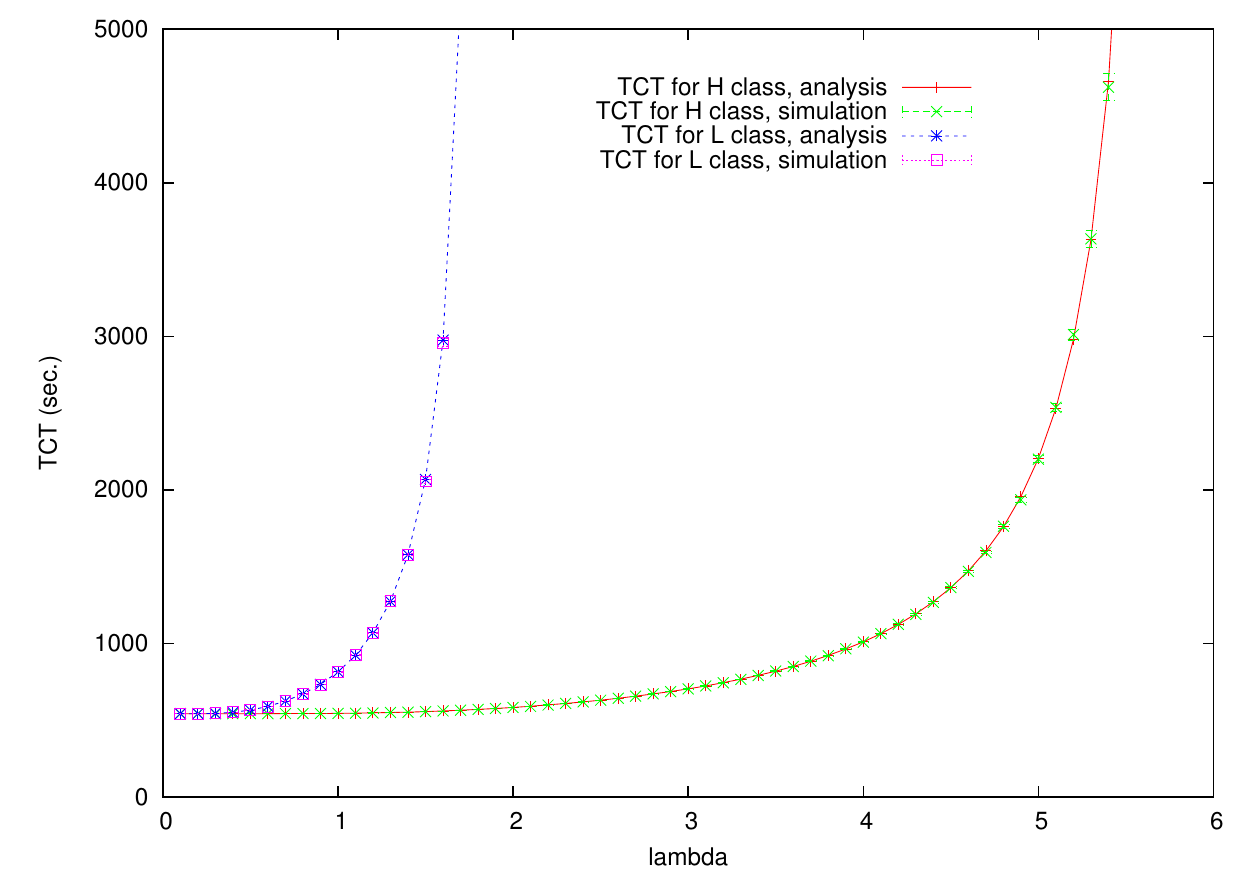}
  \caption{Comparison of analysis and simulation for the
    transaction-confirmation time: Two-priority case.}
  \label{fig:ne:comparison}
\end{figure}

In order to validate the analysis in section \ref{sec:analysis}, we
conducted discrete-event simulation experiments. The simulation model
is the same as the priority queueing one described in section
\ref{sec:analysis}. We developed a simulation program with C++, and
generated 50 samples for one estimate of the transaction-confirmation
time, calculating the 95\% confidence interval.

Figure \ref{fig:ne:comparison} represents the analytical and
simulation results of mean transaction-confirmation times for H and L
classes. Here, we set $b=1000$, and the horizontal axis is the overall
arrival rate $\lambda$, given by $\lambda=\lambda_H+\lambda_L$. We
increase $\lambda$, keeping the ratio of $\lambda_H$ to $\lambda_L$
constant. More precisely, let $\zeta$ denote the ratio of $\lambda_H$
to $\lambda_L$. From Table \ref{tab:statistics:tct}, we set $\zeta$ as
\[
\zeta = \frac{\lambda_H}{\lambda_L} =
\frac{0.90466}{0.068082} = 13.288.
\]
By using $\zeta$, $\lambda_H$ and $\lambda_L$ are described as
\[
\lambda_H = \frac{\zeta\lambda}{1+\zeta}, \quad
\lambda_L = \frac{\lambda}{1+\zeta}.
\]
With $\lambda_H$ and $\lambda_L$, we calculate $E[T_H]$ and $E[T_L]$ as the
function of $\lambda$. 

Figure \ref{fig:ne:comparison} shows overall good agreement between
the analysis and simulation for both H and L classes. Remind that
$E[T_i]$'s given in Theorem \ref{th:priority} (and hence $E[T_H]$ of
(\ref{eq:meandelay:02}) and $E[T_L]$ of (\ref{eq:meandelay:03})) are
approximations. Figure \ref{fig:ne:comparison} suggests that our
approximation analysis becomes exact when the block-generation time is
exponentially distributed.

In the following subsections, we show the numerical results calculated
by (\ref{eq:meandelay:02}) and (\ref{eq:meandelay:03}).

\subsubsection{Comparison of analysis and measurement}

\begin{table}
  \centering
  \caption{Comparison of analysis and measurement for the
    transaction-confirmation time.}
  \begin{tabular}{crrr}\hline
    Transaction Type & Arrival Rate & Measurement & Analysis \\\hline
    Classless & 0.97275 & 1,075.0 & 568.10 \\
    H         & 0.90466 & 874.13 & 562.16 \\
    L         & 0.068082 & 3,744.1 & 647.05 \\\hline
  \end{tabular}
  \label{tab:nr:comparison}
\end{table}

Next, we compare analytical results of the transaction-confirmation
time with measurement ones of Table \ref{tab:statistics:tct}.  Table
\ref{tab:nr:comparison} shows the results of measurement and analysis
for the transaction-confirmation time in three cases: classless, H
class and L class.  We calculate the transaction-confirmation time for
classless case by (\ref{eq:meannumtran:01}), while we compute $E[T_H]$
(resp.~$E[T_L]$) from (\ref{eq:meandelay:02})
(resp.~(\ref{eq:meandelay:03})). In the analytical computation, we set
$b=1750$, which is an estimate obtained from Table
\ref{tab:statistics:nt}.

In Table \ref{tab:nr:comparison}, the measurement value for classless
case is almost twice larger than the corresponding analytical one. We
also observe that discrepancies between measurement and analysis for
H and L classes are large, and that the discrepancy for L class is
significantly larger than that for H class.

First, we consider the reason of the discrepancy for classless case.
In the previous subsection, we concluded that the block-generation
time follows an exponential distribution with mean 544.095 [s]. Note
that the arrival rate of classless case is 0.97275, and hence the
system utilization $\rho$ is
\[
\rho = \lambda E[S] = 0.97275\times 544.095 = 529.27.
\]
Since the maximum block size $b$ is 1750, the system is not
overloaded. In such a situation of low utilization, a newly arriving
transaction is likely to be included in the block which is under the
current mining process.

Remind that our analytical model follows the behavior of the default
bitcoin client for updating the blockchain described in
\cite{Antonopoulos14}, that is, miners include newly arriving
transactions into the creating block as long as the resulting block
size is smaller than the maximum block size. The above comparison
result implicitly means that a newly arriving transaction is not
included in the block currently processed, but is included to the
block following the currently processed block.

This conjecture is supported by the fact that the block-generation
time follows an exponential distribution. In the underloaded
situation, the confirmation time of a newly arriving transaction
consists of the remaining generation time of the block under mining
and the generation time of the next block. Due to the memoryless
property of exponential distribution, the remaining block-generation
time also follows the same exponential distribution. This results in
that the transaction-confirmation time is almost twice larger than the
block-generation time.

Next, we consider the reason why the discrepancy between measurement
and analysis for L class is larger than that for H class. In our
analytical model, we assumed that L-class transactions in system are
served as long as the block being in service is not occupied by
H-class transactions. The large discrepancy between measurement and
analysis for L class in Table \ref{tab:nr:comparison} implies that
L-class transactions in Bitcoin system are less served than the assumed
priority queueing discipline.  As we stated in introduction, there is
little incentive for miners to build a block with transactions with small
fees. This result suggests that there exist miners who intentionally
exclude transactions with small fees from the block inclusion process.

According to the above discussion, we can conjecture that miner nodes
don't follow the behavior of the default bitcoin client, and that
there may exist miners who never include transactions with small fees
to a block.

\subsection{Mean transaction-confirmation time: classless case}

\begin{figure}[t]
  \centering
  \includegraphics[bb=0 0 410 302,width=.8\textwidth,clip]{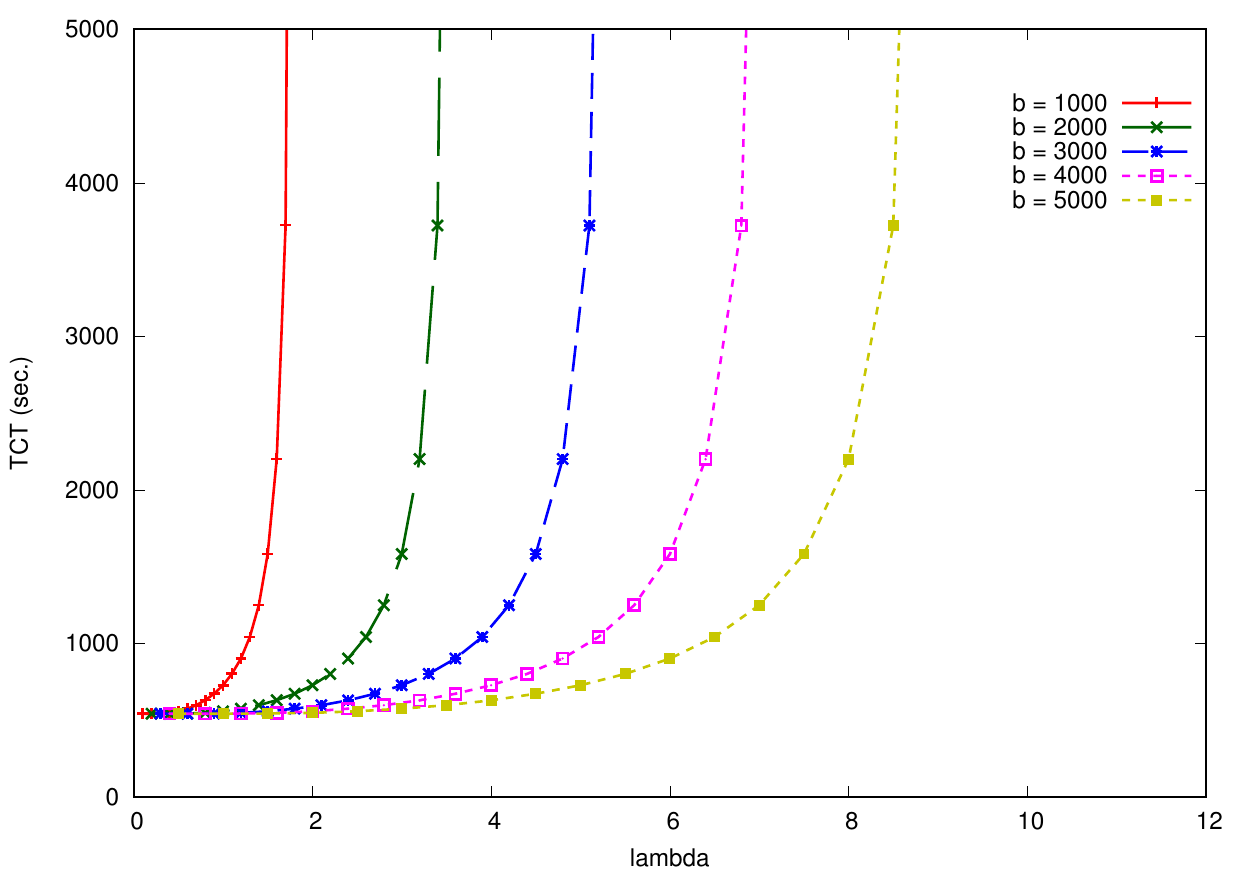}
  \caption{Mean transaction-confirmation time: classless case.}
  \label{fig:ne:lambda}
\end{figure}

In this subsection, we show the mean transaction-confirmation time for
classless case.  Figure \ref{fig:ne:lambda} represents the mean
transaction-confirmation time $E[T]$ against the overall
transaction-arrival rate $\lambda$. Here, we plot $E[T]$'s for
$b=1000$, 2000, 3000, 4000 and 5000. In this figure, $E[T]$ for each
$b$ increases from 544 s, the mean block-generation time, and grows to
infinity as $\lambda$ approaches $b/\mu$ ($=bE[S]$).

Note that the case of $b=2000$ approximately illustrates the
transaction-confirmation time under the block-size limit of 1 Mbyte. The
transaction-confirmation time rapidly increases when $\lambda$ becomes
greater than 3. Roughly speaking, the transaction-confirmation time
becomes intolerable when the number of transactions issued in one
second is greater than three. This is just the reason why the maximum
block-size limit is an important issue for the scalability of Bitcoin.

Note also that $b=3000$, 4000 and 5000 can be regarded as cases of the
maximum block size equal to 1.5 Mbytes, 2 Mbytes and 2.5Mbytes,
respectively. We can see that enlarging the maximum block size is
effective to make the transaction-confirmation time small. Even when
$b=5000$, however, the transaction-confirmation time becomes worse
around $\lambda=8$.  This result suggests that enlarging the maximum
block size is not effective for the scalability of Bitcoin.

\subsection{Mean transaction-confirmation  time:  two-priority case}

In this subsection, we investigate how the priority mechanism in
Bitcoin affects the transaction-confirmation time. We consider two
scenarios in terms of the increase in the arrival rate of
transactions. In the first scenario, $\lambda_L$ changes under a fixed
$\lambda_H$. This scenario illustrates the case in which the demand of
micro payment grows independently. In the second scenario, on the
other hand, we increase the overall transaction-arrival rate
$\lambda=\lambda_H+\lambda_L$, keeping the ratio of $\lambda_H$ to
$\lambda_L$ constant. This case corresponds to the growth of Bitcoin-user
population.

\subsubsection{Impact of increase in L-class transactions}

\begin{figure}[t]
  \centering
  \includegraphics[bb=0 0 410 302,width=.8\textwidth,clip]{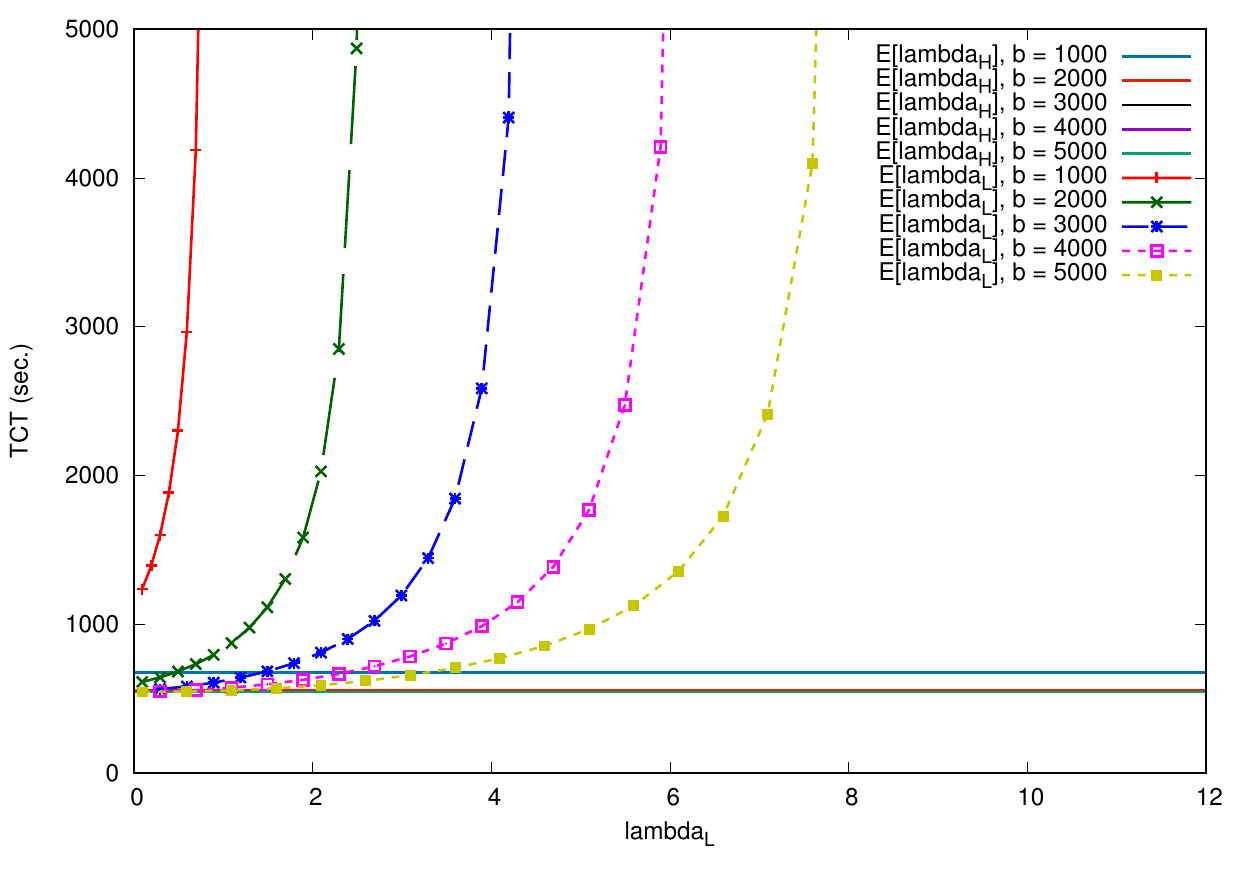}
  \caption{Mean transaction-confirmation time: two-priority
    case. ($\lambda_H=0.90466$)}
  \label{fig:ne:ex1}
\end{figure}

Figure \ref{fig:ne:ex1} represents how the mean
transaction-confirmation time is affected by the arrival rate of
L-class transactions. In this figure, $\lambda_H$ is fixed at 0.90466,
as shown in Table \ref{tab:statistics:tct}, and we plot five cases of
$b$.

In Figure \ref{fig:ne:ex1}, $E[T_L]$ for each $b$ grows exponentially
with the increase in $\lambda_L$, while $E[T_H]$'s are almost the same
and remain constant. This result shows that the priority mechanism
provides a low transaction-confirmation time for H-class transactions,
while L-class transactions are likely to suffer from a large
transaction-confirmation time when the arrival rate of L-class
transactions is high. Remind that the mean arrival rate of L-class
transactions is 0.068082, and that the current maximum block size can
be roughly approximated by $b=2000$.  Figure \ref{fig:ne:ex1}
indicates that if the arrival rate of L-class transactions becomes 30
times larger than 0.068082 ($\lambda_L\approx 2$) and the maximum
block size is limited to 1 Mbyte, the resulting confirmation time of
L-class transactions is extremely large.

\subsubsection{Growth of Bitcoin-user population}

\begin{figure}[t]
  \centering
  \includegraphics[bb=0 0 410 302,width=.8\textwidth,clip]{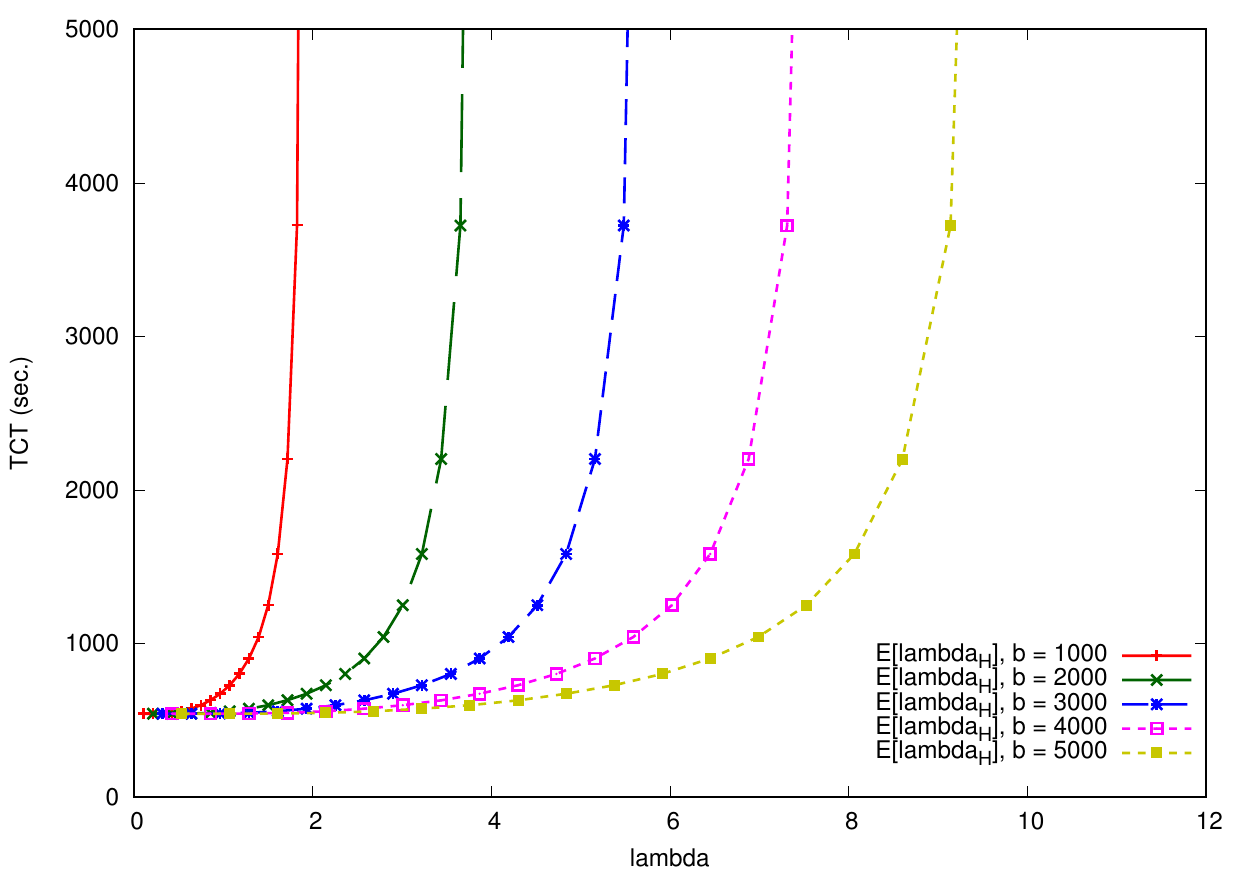}
  \caption{Mean transaction-confirmation time: high priority
    case. The ratio of $\lambda_H$ to $\lambda_L$ is fixed, and the
    overall arrival rate $\lambda$ changes.}
  \label{fig:ne:ex2:h}
\end{figure}

\begin{figure}[t]
  \centering
  \includegraphics[bb=0 0 410 302,width=.8\textwidth,clip]{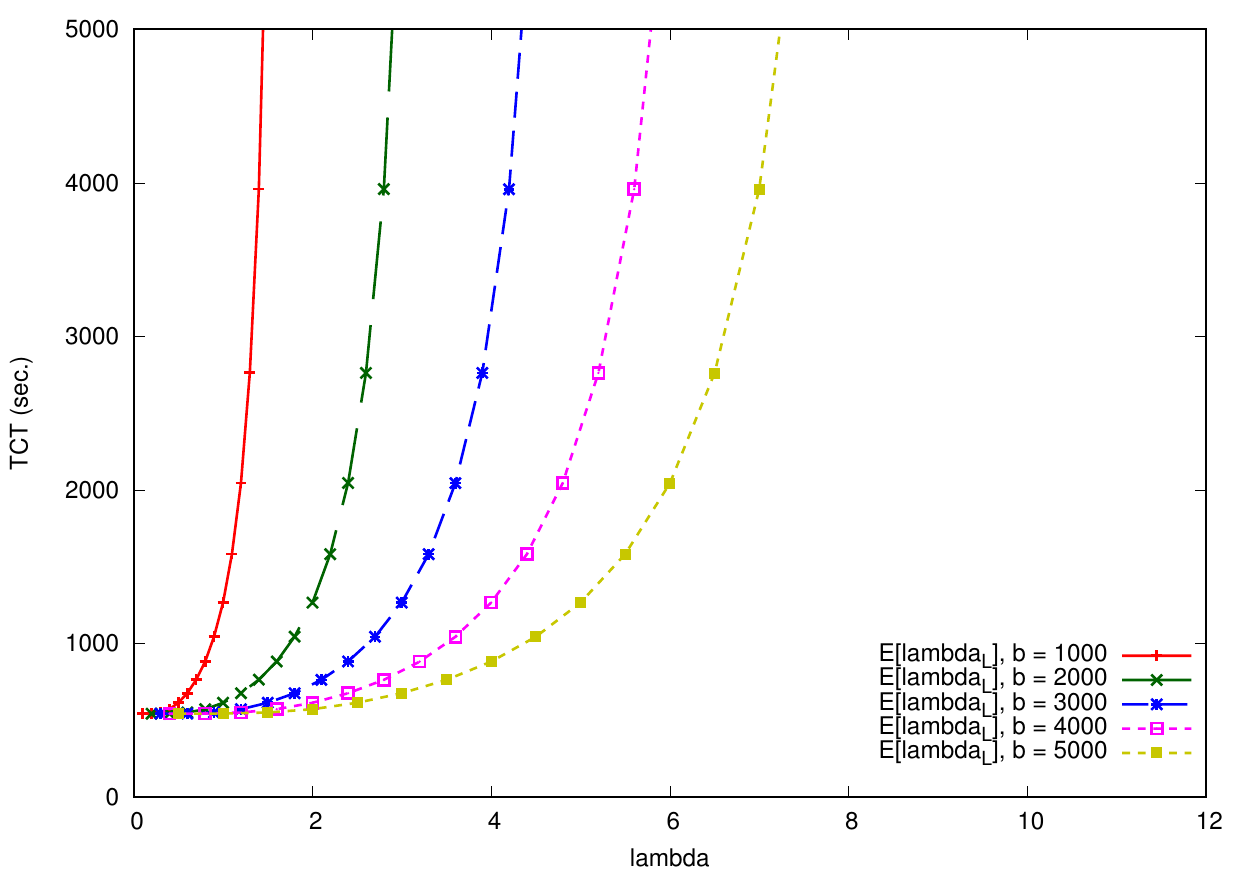}
  \caption{Mean transaction-confirmation time: low priority
    case. The ratio of $\lambda_H$ to $\lambda_L$ is fixed, and the
    overall arrival rate $\lambda$ changes.}
  \label{fig:ne:ex2:l}
\end{figure}


Figures \ref{fig:ne:ex2:h} and \ref{fig:ne:ex2:l} show how $\lambda$
affects $E[T_H]$ and $E[T_L]$, respectively. In both figures, the
horizontal axis represents $\lambda$, and we plot $E[T_H]$'s and
$E[T_L]$'s given by (\ref{eq:meandelay:02}) and
(\ref{eq:meandelay:03}) for the five cases of $b$.  In these figures,
both $E[T_L]$ and $E[T_H]$ grow exponentially with the increase in
$\lambda$. We also observe that for $b=2000$, $E[T_H]$ grows rapidly
as $\lambda$ approaches 3.  This indicates that under the current
block-size limit of 1 Mbyte, even high-class transactions suffer from
a huge confirmation time when the usage demand of Bitcoin grows three
times larger than the current situation.

Figures \ref{fig:ne:ex2:h} and \ref{fig:ne:ex2:l} also show that
increasing the maximum block size is effective to mitigate the rapid
growth of the transaction-confirmation time. When $b=5000$, the growth
of $E[T_H]$ in Figure \ref{fig:ne:ex2:h} is slow, however, $E[T_H]$
rapidly increases around $\lambda=9$. This result indicates that
increasing the maximum block size is not a fundamental solution for
the scalability of Bitcoin.

\section{Conclusion}
\label{sec:conclusion}

In this paper, we analyzed the transaction-confirmation time for
Bitcoin by queueing theory. We modeled the transaction-confirmation
process as a single-server queue with batch service and priority
mechanism. Assuming that the priority of a transaction depends only on
its input, we derived the mean confirmation time for transactions of
each priority class. Numerical examples showed that for the maximum
block size of 1 Mbyte, transactions with small fees suffer from an
extremely large confirmation time if the arrival rate of transactions
whose fee is smaller than 0.0001 BTC becomes four times larger than
the current arrival rate. We also found that enhancing the maximum
block size is not an effective way to mitigate the
transaction-confirmation time. Further study is needed for the
scalability of Bitcoin.

\section*{Acknowledgment}

The first author would like to thank Prof.~Tetsuya Takine of Osaka
University for his valuable comment on the analysis of the queueing
model in the paper.  This research was supported in part by SCAT
Foundation, and Japan Society for the Promotion of Science under
Grant-in-Aid for Scientific Research (B) No.~15H04008.

\appendix

\section{Proofs of Theorems in Priority Queueing Analysis}
\label{app:sec:analysis}

\subsection{Proof of Theorem 1}
\label{app:sec:proof:sub:theorem1}

When $\lambda E[S] < b$ holds, the system is stable and hence limiting
probabilities exist. Letting
$P_n(x) = \lim_{t\rightarrow\infty}P_n(x,t)$ and
$P_0=\lim_{t\rightarrow\infty}P_0(t)$, we obtain from the assumptions
\begin{eqnarray}
  \lambda P_0 &=& \sum_{k=1}^b \int_0^\infty P_k(x) \xi(x) \rmd x,
     \label{eq:balance:01} \\ 
  \frac{\rmd}{\rmd x}P_n(x) &=& -\{\lambda + \xi(x)\}P_n(x) + \lambda
  P_{n-1}(x), \quad n=2,3,\ldots,   \label{eq:balance:02}\\
  \frac{\rmd}{\rmd x}P_1(x) &=& -\{\lambda +
                  \xi(x)\}P_1(x).   \label{eq:balance:03} 
\end{eqnarray}
Intuitively, (\ref{eq:balance:01}) is a balance equation in which the
exiting rate from state 0 is equal to the entering rate into the same
state. The first term in the right-hand side (r.h.s.) of
(\ref{eq:balance:02}) is derived from the event that the number of
transactions does not change during a small time interval, while the
second term is yielded from the event that a transaction arrives at
the system with $n$ transactions. We derive (\ref{eq:balance:03}) in a
similar manner.

We also have the following boundary conditions
\begin{eqnarray}
  P_n(0) &=& \int_0^\infty P_{n+b}(x)\xi(x) \rmd x,\quad
             n=2,3,\ldots,   \label{eq:boundary:01}\\
  P_1(0) &=& \int_0^\infty P_{1+b}(x)\xi(x) \rmd x + \lambda P_0.
             \label{eq:boundary:02}
\end{eqnarray}
Note that the left-hand side of (\ref{eq:boundary:01}) is the
probability that there exist $n$ transactions in system at the
beginning of the service time. This event occurs just after the
service completion in the state with $n+b$ transactions. (Remind that
$b$ transactions are served simultaneously when the number of
transactions in system is greater than or equal to $b$.)  The equation
of (\ref{eq:boundary:02}) can be derived in a similar manner, noting
that the service with one transaction starts when a transaction
arrival occurs at system in idle (the second term in the r.h.s.~of
(\ref{eq:boundary:02})).

The normalizing condition is given by
\begin{equation}
  \label{eq:normal:01}
  P_0 + \sum_{n=1}^\infty \int_0^\infty P_n(x) \rmd x = 1.
\end{equation}

We define the following probability generating functions (pgf's)
\begin{eqnarray*}
  P(z;x) &=& \sum_{n=1}^\infty P_n(x) z^n,\\
  P(z) &=& P_0 + \int_0^\infty P(z;x) \rmd x.
\end{eqnarray*}
Multiplying (\ref{eq:balance:02}) by $z^n$ and (\ref{eq:balance:03})
by $z$, and summing over $n=1,2,\ldots$, we obtain
\begin{equation}
  \label{eq:pzx:01}
  P(z;x) = P(z;0)\{1-G(x)\}\exp\{-\lambda (1-z)x\}.
\end{equation}

From the boundary conditions (\ref{eq:boundary:01}) and
(\ref{eq:boundary:02}), we also have
\begin{equation}
  \label{eq:boundarypgf:01}
  P(z;0) ={\sum_{k=1}^b(z^{b+1}-z^k) \over z^b - G^*(\lambda - \lambda z)}
     \cdot \int_0^\infty P_k(x) \xi(x) \rmd x,
\end{equation}
where $G^*(s)$ is the LST of $G(x)$ and given by
\[
G^*(s) = \int_0^\infty e^{-sx} \rmd G(x).
\]
For notational simplicity, we define $\alpha_k$ ($k=1,2,\ldots,b$) as
\[
\alpha_k = \int_0^\infty P_k(x) \xi(x) \rmd x.
\]

From Rouche's theorem (see, for example, \cite{Takagi91}), it is
shown that the equation
\begin{equation}
  \label{eq:rouche:01}
  z^b - G^*(\lambda -\lambda z) = 0,
\end{equation}
has $b$ roots inside $|z|=1+\epsilon$ for a small real number
$\epsilon>0$. One of them is $z=1$. Let $z_m^*$ ($m=1,2,\ldots,b-1$)
denote the $m$-th root of (\ref{eq:rouche:01}) different from 1. Hence,
from (\ref{eq:boundarypgf:01}), we have the following $b-1$ equations
\begin{equation}
  \label{eq:rouche:02}
  \sum_{k=1}^b \{(z_m^*)^{b+1} - (z_m^*)^k\} \cdot \alpha_k = 0, \quad
  m=1,2,\ldots, b-1. 
\end{equation}

From (\ref{eq:pzx:01}), we obtain
\begin{eqnarray}
  \int_0^\infty P(z;x) \rmd x &=& \int_0^\infty 
     P(z;0)\{1-G(x)\}\exp\{-\lambda (1-z)x\} \rmd x \nonumber\\
     &=& P(z;0) \frac{1-G^*(\lambda -\lambda z)}{\lambda - \lambda z}.
  \label{eq:pgfint:01}
\end{eqnarray}
From (\ref{eq:pgfint:01}), $P(z)$ is yielded as
\begin{equation}
  \label{eq:pgfpz:01}
  P(z) = P_0 + P(z;0) \frac{1-G^*(\lambda -\lambda z)}{\lambda - \lambda z}.
\end{equation}

Substituting $z=1$ into (\ref{eq:boundarypgf:01}) yields
\begin{equation}
  \label{eq:pgfpz:02}
  P(1;0) = {\sum_{k=1}^b (b+1-k)\alpha_k \over b-\lambda E[S]}.
\end{equation}
Note that (\ref{eq:pgfpz:02}) holds if the following stability
condition holds.
\begin{equation}
  \label{eq:pgfpz-stability}
  \lambda E[S]< b.
\end{equation}

Noting that $P(1)=1$, we obtain from (\ref{eq:balance:01}),
(\ref{eq:pgfpz:01}) and (\ref{eq:pgfpz:02})
\begin{equation}
  \label{eq:pgfpz:03}
  \sum_{k=1}^b \left\{
    {(b+1-k)E[S] \over b-\lambda E[S]} + \frac{1}{\lambda}\right\}
    \cdot\alpha_k = 1.
\end{equation}
From (\ref{eq:rouche:02}) and (\ref{eq:pgfpz:03}), $\alpha_k$'s are
uniquely determined.

Using $\alpha_k$'s, (\ref{eq:balance:01}) and
(\ref{eq:boundarypgf:01}) can be rewritten as
\[
P_0 = \frac{1}{\lambda}\sum_{k=1}^b \alpha_k,\quad
P(z;0) = {\sum_{k=1}^b(z^{b+1} - z^k) \alpha_k \over z^b -G^*(\lambda -
  \lambda z)}.
\]
Substituting the above expressions into (\ref{eq:pgfpz:01}) yields
\begin{equation}
  \label{eq:pgfpz:04}
  P(z) = \frac{1}{\lambda}\sum_{k=1}^b \alpha_k +
     {\sum_{k=1}^b(z^{b+1} - z^k) \alpha_k \over z^b -G^*(\lambda -
  \lambda z)}\cdot {1-G^*(\lambda - \lambda z) \over \lambda -\lambda z}.
\end{equation}

The mean number of transactions in the system $E[N]$ is given by
\begin{eqnarray*}
  \lefteqn{E[N] = \left(\frac{\rmd}{\rmd z}P(z)\right)_{z=1}} \\
     &=& {1\over 2\lambda (b-\lambda E[S])}  
      \left(\rule{0pt}{18pt}
       \sum_{k=1}^b \alpha_k\left[\rule{0pt}{14pt}
         b(b-1)+\{(b+1)b-k(k-1)\}\lambda E[S] \right.\right.\\
     & & \left.\left. +(b-k)\lambda^2E[S^2]\rule{0pt}{14pt}\right]
        - \lambda\left\{b(b-1)-\lambda^2
          E[S^2]\right\}\rule{0pt}{18pt}\right).
\end{eqnarray*}

Let $T$ denote the sojourn time of a transaction. Note that in Bitcoin
case, $T$ is the transaction-confirmation time.  From Little's
theorem, the mean sojourn time of a transaction $E[T]$ is yielded as
\begin{eqnarray*}
  \lefteqn{E[T] = \frac{E[N]}{\lambda}}\nonumber\\
     &=& {1\over 2\lambda^2 (b-\lambda E[S])}  
      \left(\rule{0pt}{18pt}
       \sum_{k=1}^b \alpha_k\left[\rule{0pt}{14pt}
         b(b-1)+\{(b+1)b-k(k-1)\}\lambda E[S] \right.\right.\nonumber\\
     & & \left.\left. +(b-k)\lambda^2E[S^2]\rule{0pt}{14pt}\right]
        - \lambda\left\{b(b-1)-\lambda^2
          E[S^2]\right\}\rule{0pt}{18pt}\right).
\end{eqnarray*}
Since $E[T]$ is a function of $\lambda$, we define $f(\lambda)\equiv
E[T]$ for the following subsection.

\subsection{Proof of Priority Queueing Analysis}
\label{app:sec:proof:sub:priority}

Consider a sample path in which a low-priority transaction arrives at
the system in idle and starts a busy period. Consider also the other
sample path in which a high-priority transaction arrives at the system
in idle and starts a busy period. In our model, note that the elapsed
service time of high-priority transaction in the former sample path is
smaller than that in the latter one. This implies that our priority
queueing model is not work conserving.

When the system utilization $\sum_{i=1}^c \lambda_i E[S]$ is large,
however, the busy period becomes large and the idle state rarely
occurs.  In such high-utilization environment, the event that an
low-priority transaction starts a busy period rarely occurs.

Assuming that the system is work conserving \cite{Wolf89}, we have
\begin{equation}
  \label{eq:priority:01}
  f(\overline{\lambda}_c) = \sum_{k=1}^c
  \frac{\lambda_k}{\overline{\lambda}_c} E[T_k],
\end{equation}
where $E[T_k]$ is the sojourn time of class $k$ transactions.  Since
class 1 transactions are served similarly to the batch service
analyzed in the previous subsection, $E[T_1]$ is given by
\begin{equation}
  \label{eq:meandelay:00}
  E[T_1]=f(\lambda_1).
\end{equation}

Note that for $i<j$, any class-$j$ transactions don't affect the
service of class-$i$ transactions. In other words, $T_i$ is
independent of transactions whose priority class is lower than
$i$, and hence (\ref{eq:priority:01}) holds not only $c$ but also
$i=2,3,\ldots,c-1$. This yields
\begin{eqnarray*}
  f(\overline{\lambda}_i) &=& \sum_{k=1}^i
  \frac{\lambda_k}{\overline{\lambda}_i} E[T_k]\\
  &=& \sum_{k=1}^{i-1}
  \frac{\lambda_k}{\overline{\lambda}_i} E[T_k]
  + \frac{\lambda_i}{\overline{\lambda}_i} E[T_i].
\end{eqnarray*}
We then obtain
\begin{equation}
  \label{eq:meandelay:01}
  E[T_i] = \frac{1}{\lambda_i}
  \left(\overline{\lambda}_i f(\overline{\lambda}_i) -
    \sum_{k=1}^{i-1} \lambda_k E[T_k]
  \right), \quad i=2,3,\ldots, c.
\end{equation}
Note that $E[T_i]$'s can be calculated recursively by
(\ref{eq:meandelay:00}) and (\ref{eq:meandelay:01}).

\section{Block-Generation Time Distribution}
\label{app:sec:block-gener-time}

In this section, we prove that the block-generation time follows an
exponential distribution. 

Remind that each miner node tries to solve the mathematical problem
based on a cryptographic hash algorithm. This problem consists of
calculating a hash of the block being formed and adjusting a nonce
word such that the resulting hash value is smaller than or equal to a
targeted value called difficulty \cite{TS16}.  The number of nonce
words the miner tries is tremendously huge, making the mathematical
problem too difficult. Here, we assume that the number of nonce words
is finite and equal to $M$.

When a miner tries one nonce word and finds it incorrect, the miner
immediately tries the other word and never tries the same nonce again.
We can model the mining process as the following urn model without
replacement. That is, we have an urn containing $M$ balls, of which
$M-1$ are white and one is red. The red ball is a winner. One ball is
withdrawn from the urn at a time, and then it is removed from the urn
without replacement.  In this setting, the probability that the red
ball is drawn at $k$th trial is $1/M$, a discrete-uniform
distribution. If one trial is performed at a unit time, the
probability that the red ball is drawn at time $k$ is also given by
$1/M$.

Suppose that there exist $n$ miner nodes. Without loss of generality,
we assume that the number of winning nonce words is one and that the
number of nonce words is $M$.  Let $Y_i$ ($i=1,2,\ldots,n$) denote the
time at which miner $i$ finds a winning nonce word. We assume $Y_i$'s
are i.i.d. We define the block-generation time as $L_n$. 
Then we have
\[
L_n = \min\{Y_1, Y_2,\ldots,Y_n\}.
\]

For simplicity, we assume $Y_i$ follows a continuous-uniform
distribution $U(0,M)$, that is,
\[
\Pr\{Y_i \leq x\} = \left\{
  \begin{array}{cl}
    x/M, & 0\leq x\leq M, \\
    0, & \mbox{others.}
  \end{array}
\right.
\]
Then, the distribution of $L_n$ is yielded as
\begin{align*}
  \Pr\{L_n\leq x\} &= \Pr\{\min(X_1,\ldots,X_n) \leq x\} \\
     &= 1-\Pr\{\min(X_1,\ldots,X_n) > x\} \\
     &= 1-\left(1-\frac{x}{M}\right)^n.
\end{align*}

Now consider a limit distribution of $(L_n-b_n)/a_n$ for sequences of
constants $\{a_n>0\}$ and $b_n$. In extreme value theory, it is known
that the distribution of $(L_n-b_n)/a_n$ for the minimum of $X_i$'s
converges to a Weibull distribution when $X_i$ follows uniform
distribution (\cite{KN00} p.~59, Table A.1).

For $0\leq z\leq n$, setting $a_n=1/n$ and $b_n=0$ yields
\begin{align*}
\Pr\left\{ \frac{L_n-b_n}{a_n}\leq z \right\} &= 1-\left\{
  1-\frac{(z/M)}{n} \right\}^n \\
   &\rightarrow 1 - e^{-z/M}, \quad n\rightarrow\infty.
\end{align*}
From this result, for a large $n$, we can approximate the distribution
of $L_n$ by 
\[
\Pr\{L_n\leq x\} \approx 1-\exp\{-(n/M)x\}.
\]
This result implies that $L_n$ follows an exponential distribution
when $n$ is large.

Figure \ref{fig:ne:fitting} in subsection
\ref{sec:numericalexamples:sub:block-generation-time} shows a good
agreement between exponential distribution and measured data.
According to \cite{bitnode}, the number of miner nodes is about 5,700,
and hence this number is large enough so that the block-generation
time is well approximated by exponential distribution.







\end{document}